\documentclass[11pt]{article}

\usepackage[utf8]{inputenc}
\usepackage[T1]{fontenc}
\usepackage{graphicx,amsthm}
\usepackage{amsmath}
\usepackage{fancyhdr}
\usepackage{amsfonts}
\usepackage{amssymb}
\usepackage{enumerate}
\usepackage{tikz}
\usepackage{subfigure}

\usepackage{xcolor}
\usepackage{graphicx}
\usepackage{tabularx}
\usepackage[center]{caption}
\usepackage{array}

\usepackage{multirow}

        \theoremstyle{plain}

\def\ds{\displaystyle}
\newcommand{\R}{{\mathbb{R}}}

\newcommand{\divv}{\nabla_v \cdot}
\newcommand{\divx}{\nabla_x \cdot}

\newcommand{\dt}{\partial_t}

\newcommand{\p}{\partial}

\newcommand{\Kn}{\mathrm{Kn}\thinspace}
\newcommand{\la}{\left\langle}
\newcommand{\ra}{\right\rangle}

\newcommand{\cint}[1]{\langle #1 \rangle}

\def\eps{\varepsilon}
\usepackage{amsmath}
\usepackage{amsthm}

        \theoremstyle{plain}
        
        \newtheorem{proposition}{Proposition}[section]
        \newtheorem{prop}{Property}[section]

        \theoremstyle{remark}
        
        \theoremstyle{remark}
        
        \theoremstyle{remark}

 \newcommand{\correct}[1]{{#1}}

\begin{document}

\begin{center}
{\bf BGK and Fokker-Planck models of the Boltzmann equation for gases
  with discrete levels of vibrational energy}

\vspace{1cm}
J. Mathiaud$^{1,2}$,  L. Mieussens$^2$

\bigskip
$^1$CEA-CESTA\\
15 avenue des sabli\`eres - CS 60001\\
33116 Le Barp Cedex, France\\
{ \tt(julien.mathiaud@cea.fr)}

\bigskip
$^2$Univ. Bordeaux, Bordeaux INP, CNRS, IMB, UMR 5251, F-33400 Talence, France.\\
{ \tt(Luc.Mieussens@math.u-bordeaux.fr)}

\end{center}

\begin{abstract}
We propose two models of the Boltzmann equation (BGK and Fokker-Planck
models) for rarefied flows of diatomic gases in vibrational non-equilibrium. These
models take into account the discrete repartition of vibration energy
modes, which is required for high temperature flows, like for
atmospheric re-entry problems. We prove that these models satisfy
conservation and entropy properties (H-theorem), and we derive their
corresponding compressible Navier-Stokes asymptotics.
\end{abstract}

 \bigskip


%----------------------------------%
%----        CHAPITRE I        ----%
%----------------------------------%
\section{Introduction}

Numerical simulation of atmospheric reentry flows requires to solve
the Boltzmann equation of Rarefied Gas Dynamics. The standard method
to do so is the Direct Simulation Monte Carlo (DSMC) method
\cite{bird,BS_2017}, which is a particle stochastic method. However,
it is sometimes interesting to have alternative numerical methods,
like, for instance, methods based on a direct discretization of the
Boltzmann equation \correct{(see \cite{dimarco_pareschi_2014})}. This is hardly possible
for the full Boltzmann equation (except for monatomic gases,
see~\cite{luc_2014}), since this is still much too computationally
expensive for real gases. But BGK like model equations~\cite{bgk} are
very well suited for such deterministic codes: indeed, their
complexity can be reduced by the well known reduced distribution
technique~\cite{Chu_1965}, which leads to intermediate models between
the full Boltzmann equation and moment
models~\cite{Struchtrup_moment_book}. The Fokker-Planck
model~\cite{cercignani} is another model Boltzmann equation that can
give very efficient stochastic particle methods, see~\cite{Grj2011}.

These model equations have already been extended to polyatomic
gases, so that they can take 
into account the internal energy of rotation of gas molecules. They
contains correction terms that lead to correct transport coefficients: the ESBGK or
Shakhov's models~\cite{Holway,esbgk_poly,S_model}, and the cubic
Fokker-Planck and
ES-Fokker-Planck~\cite{Grj2011, Grj2013,Mathiaud2016,Mathiaud2017}.

For high temperature flows, like in space reentry problems, the
vibrational energy of molecules is activated, and has a significant
influence on energy transfers in the gas flow. It is therefore
interesting to extend the model equations to take this
vibrational modes into account. Several extended BGK models have been
recently proposed to do so, for
instance~\cite{RS_2014,WYLX_2017,ARS_2017,KKA_2019}, and a recent
Fokker-Planck model has been proposed earlier in~\cite{Grj2013}.
% PHYSICS OF FLUIDS 26, 052001 (2014) Capturing non-equilibrium phenomena in rarefied
% polyatomic gases: A high-order macroscopic model
% Behnam Rahimia) and Henning Struchtrupb)
% -> BGK continu en energie interne, 2 temps de relax
% 2017 : Unified gas-kinetic scheme for diatomic molecular flow with translational, rotational, and vibrational modes
% Zhao Wang a, Hong Yan a,b,∗, Qibing Li c, Kun Xu d
% --> BGK (extension Shakhov) avec vibra, 3 temps relax, mais continu (et modèle réduit)
% 2017 Rational Extended Thermodynamics of a Rarefied Polyatomic Gas with Molecular Relaxation Processes
% Takashi Arima1, Tommaso Ruggeri2, and Masaru Sugiyama3 : 
% BGK à temps relax, énergie continue, pas de modèle réduit
% aoki 2019 : ESBGK avec delta(T), continu en energie, peut tenir compte
% vibra, pas de modèle réduit

All these models assume a continuous vibrational energy
repartition. However, while transitional and rotational energies in
air can be considered as continuous for temperature larger than $1$K
and $10$K, respectively, vibrational energy can be considered as
continuous only for much larger temperatures ($2000$K for oxygen and
$3300$K for nitrogen). For flows up to $3000$K around reentry
vehicles, the discrete levels of vibrational energy must be used~\cite{anderson}. It
seems that that the only BGK model that allows for this discrete
repartition is the model of Morse~\cite{Morse_1964}.
% Kinetic Model for Gases with Internal Degrees of Freedom
% T. F. Morse : modèle BGK avec énergie interne discrete, avec 2 maxw et
% 2 temps de relax
% Such an approach has probably not been used for simulation  due to the
% difficulty to take into account a very large number of energy modes
% into the discretization.

In this paper, we consider a simpler version of this Morse BGK model
for vibrating gases that allows for a discrete vibrational energy. We
show that the complexity of this model can be reduced with the reduced
distribution technique so that the discrete vibrational energy is
eliminated. What is new here is that this construction allows us to
prove that the corresponding reduced model satisfies the
H-theorem. Moreover, the model is shown to give macroscopic Euler and
Navier-Stokes equations in the dense regime, with temperature
dependent heat capacities, as expected. This means that the reduced
model is a good candidate for its implementation in a deterministic
simulation code. \correct{Note that with this reduction, only higher order moments with respect to the vibration energy variable are lost: the macroscopic quantities of interest like pressure, temperature, and heat flux, are the same as in the non-reduced model. Moreover, since the reduced variable is not the velocity, this reduction does not require any assumption or special geometries}. 

An equivalent reduced
Fokker-Planck model is also proposed, that has the same
properties. However, this model is not based on a non-reduced model,
since we are not able so far to define diffusion process for the
discrete vibrational energy. Up to our knowledge, this is the first time
such a Fokker-Planck model for vibration energy is proposed.

Our paper is organized as follows. In section~\ref{sec:kinetic}, we
present the kinetic description of a gas with vibrating molecules, and we discuss the
mathematical properties of the reduced distributions that will be used for our
models. Our BGK and Fokker-Planck models are presented in
sections~\ref{sec:BGK} and~\ref{sec:FP}, respectively. In
section~\ref{sec:CE}, the hydrodynamic limits of our models, obtained by a
Chapman-Enskog procedure, are discussed. Section~\ref{sec:conclusion}
gives some perspectives of this work.

\section{Kinetic description of a vibrating diatomic gas}
\label{sec:kinetic}

\subsection{Distribution function and local equilibrium}
We consider a diatomic gas. We define $f(t,x,v,\eps,i)$ the mass
density of
molecules with position $x$, velocity $v$, internal energy $\eps$, and
in the $i$th vibrational energy level, such that the corresponding vibrational
energy is $iRT_0$, where $T_0 = h\nu/k$ is a characteristic vibrational
temperature of the molecule ($h$ and $k$
are the Planck and Boltzmann constant, while $\nu$ is the fundamental
vibrational frequency of the molecule).
 
The corresponding local equilibrium
distribution is defined by (see \cite{bird})
\begin{equation}
 M_{vib}[f](v,\eps,i) =\frac{\rho}{\sqrt{2\pi RT}^3}  \frac{ 1-e^{-T_0/T} }{RT} \exp\left(-\frac{\frac12|u - v|^2+\eps+iRT_0}{ RT}\right) .      
 \label{maxvib}
\end{equation}
Here, $\rho$ is the mass density of the gas, $T$ its temperature of
equilibrium and $u$ its mean velocity, defined below.

\subsection{Moments and entropy}
The macroscopic quantities are defined by moments of $f$ as follows: 
\begin{equation}  \label{eq-moments}
  \rho = \la f\ra_{v,\eps,i}, \qquad
 \rho u = \la vf\ra_{v,\eps,i}, \qquad 
\rho e =  \la \left(\frac12|v-u|^2+\eps+iRT_0 \right) f  \ra_{v,\eps,i}, 
\end{equation}
where we use the notation $ \la \phi
\ra_{v,\eps,i}=\sum_{i=0}^{\infty} \iint \phi(t,x,v,\eps,i)\, dv d\eps$ for any function $\phi$.

With standard Gaussian integrals and summation formula, it is easy to find that the moments
of the equilibrium $M_{vib}[f]$ satisfy:
\begin{equation*}
 \la M_{vib}[f] \ra_{v,\eps,i} = \rho, \qquad 
  \la v M_{vib}[f]  \ra_{v,\eps,i} = \rho u. %, \qquad
% \la (\frac12(v-u)^2+\eps+iRT_0)M_{vib}[f] \ra_{v,\eps,i}= \rho e .
\end{equation*}
At equilibrium, we can define the following energies of translation,
rotation, and vibration
\begin{align}
& \rho e_{tr}(T) = \la (\frac12(v-u)^2) M_{vib}[f]\ra_{v,\eps,i}=\frac32 \rho RT, \label{eq-etr}\\
& \rho e_{rot}(T) = \la \eps M_{vib}[f]\ra_{v,\eps,i}= \rho RT, \label{eq-rot} \\
& \rho e_{vib}(T) = \la (iRT_0 ) M_{vib}[f] \ra_{v,\eps,i}=
\rho\frac{RT_0}{e^{T_0/T}-1}
 = \frac{\delta(T)}{2}\rho RT,   \label{eq-evib}
\end{align}
where the number of degrees of freedom of vibrations is
\begin{equation}  \label{eq-deltaT}
\delta(T)=\frac{2T_0/T}{e^{T_0/T}-1},
\end{equation}
which is a non integer and
temperature dependent number, while the number of degrees of freedom is $3$ for
translation and $2$ for rotation.

The temperature $T$ is defined so that $M_{vib}[f]$ has the same
energy as $f$: 
\begin{equation*}
  \la (\frac12(v-u)^2+\eps+iRT_0)M_{vib}[f] \ra_{v,\eps,i}= \rho e,
\end{equation*}
which gives the non linear implicit definition of $T$: 
\begin{equation}  \label{eq-defT}
% = \frac{5}{2} R T + e_{vib}(T) \\
e  = \frac{5+\delta(T)}{2}RT. 
\end{equation}
Since the function  $T\rightarrow e$ is   monotonic, $T$ is
uniquely defined by~\eqref{eq-defT}. Moreover, note that $\delta(T)$
is necessarily between 0 and 2, which means that vibrations add at
most two degrees of freedom.

Finally, the entropy ${\cal H}(f)$ of $f$ is defined by ${\cal  H}(f)=   \la f\log f \ra_{v,\eps,i}.$

\subsection{Reduced distributions}
For computational efficiency, it is interesting to define marginal, or
reduced, distributions $F$ and $G$ by 
\begin{equation*}
F(t,x,v,\eps)=\sum_i f(t,x,v,\eps,i), \quad \text{ and } 
\quad  G(t,x,v,\eps)=\sum_i iRT_0 f(t,x,v,\eps,i).
\end{equation*}
The macroscopic variables defined by $f$ can be obtained through $F$
and $G$ only, as it is shown in the following proposition \correct{ by integrating with respect to $v$ and $\eps$ and using the definition~(\ref{eq-moments}) of the moments}.
\begin{proposition}[Moments of the reduced distributions]
The macroscopic variables $\rho$, $u$, and $e$, of $f$, defined
by~\eqref{eq-moments}, satisfy
\begin{equation}  \label{eq-mtsred} 
 \rho = \la F \ra_{v,\eps},\qquad
 \rho u = \la v F \ra_{v,\eps},\qquad
 \rho e = \la (\frac12(v-u)^2+\eps )  F\ra_{v,\eps}+\la G\ra_{v,\eps}.
\end{equation}
where we use the notation $ \la \psi
\ra_{v,\eps}=\iint \psi(t,x,v,\eps)\, dv d\eps$ for any function $\psi$.
\end{proposition}

This reduction procedure can be extended to the entropy functional as
follows. First, to simplify the following relations, we use the
notation $f_i(v,\eps)$ for $f(v,\eps,i)$. Then, we define the
reduced entropy by
  \begin{equation}\label{eq-HFG} 
    \begin{split}
& \mathcal{H}(F,G)= \la  H(F,G)\ra_{v,\eps}, \text{ where } \\
& H(F,G)=\inf_{f>0}\left\{ \sum_i f_i\log  f_i \quad
  \text{such that}
  \quad   \sum_i f_i = F, \quad 
 \sum_i iRT_0 f_i  = G \right\}.
    \end{split}
\end{equation}
In other words, for a given couple of reduced distributions $(F,G)$,
we define the (non reduced) distribution that minimizes the marginal
entropy $ \sum_i f_i\log  f_i$ among all the distributions that have
the same marginal distributions $F$ and $G$. Then the reduced entropy
is the integral with respect to $v$ and $\eps$ of the corresponding
marginal entropy.

Now it is possible to \correct{represent} this reduced entropy as a function of $F$
and $G$ only, as it is shown in the following proposition.
\begin{proposition}[Entropy]
The reduced entropy $\mathcal{H}(F,G)$ defined by~(\ref{eq-HFG}) is
\begin{eqnarray}
\mathcal{H}(F,G)= \la F\log(F)+ F\log\left(\frac {RT_0F}{RT_0F+G}\right)+\frac{G}{RT_0}\log\left(\frac {G}{RT_0F+G}\right) \ra_{v,\eps}.
\end{eqnarray}
\end{proposition}
\begin{proof}
  The set $\left\{f>0 \text{ such that } \sum_i f_i=F, \quad \sum_i iRT_0 f_i=G \right\}
  $ is clearly convex, so that we can use a Lagrangian multiplier
  approach by finding  \correct{if there exists a minimum} of the function $\mathcal{L}$ defined
  through :
\begin{eqnarray*}
\mathcal{L}(f,\alpha,\beta)=\sum_i f_i\log  f_i -\alpha \left( \sum_i
  f_i-F \right) 
-\beta  \left(\sum_i iRT_0 f_i-G \right),
\end{eqnarray*}
where $\alpha$ and $\beta$ are real numbers and \correct{$\sum_i f_i\log  f_i$ is a convex function of $f$} .
\correct{The functional $\mathcal{L}$ is convex but no longer strictly convex.  A minimum of  $\mathcal{H}(F,G)$ necessarily } satisfies $\frac{\partial \mathcal{L}}{\partial f} =
0$, and it is easy to deduce that $f$ can be written
$f_i(v,\eps)=A\exp\left(-iBRT_0\right)$, where $A:=A(v,\eps)$ and
$B:=B(v,\eps)$ are functions that are still to be determined. 

\correct{The linear constraints give:
\begin{align*}
& F=\sum_i f_i=\frac A {1-\exp\left(-BRT_0\right)},\\
& G=\sum_i  iRT_0f_i=\frac{ART_0\exp\left(-BRT_0\right)}{\left(1-\exp\left(-BRT_0\right)\right)^2},
\end{align*}
where we have used the property $iRT_0 f_i  = -\frac{df_i}{dB} $ that comes from $f_i=A\exp\left(-iBRT_0\right)$.}
Solving this linear system gives
\begin{equation*}
A=\frac{RT_0F^2}{RT_0F+G}, \qquad 
B=\frac 1{RT_0}\log\left(1+\frac{RT_0F}{G}\right).
\end{equation*}
so that 
\begin{eqnarray}
H(F,G)=F\log(F)+ F\log\left(\frac {RT_0F}{RT_0F+G}\right)+\frac{G}{RT_0}\log\left(\frac {G}{RT_0F+G}\right).
\end{eqnarray}
A final integration with respect to $v$ and $\eps$ gives the final result.
\end{proof}

{The
following proposition gives useful differential
properties of the reduced entropy functional. 
\begin{proposition}[Properties of $H$] \label{diffentrop}

\begin{enumerate}
\item  The partial derivatives of $H$ computed at $(F,G)$ are:
\begin{equation}  \label{eq-partialH}
  D_1H(F,G)=1+\log\left(\frac{RT_0F^2}{RT_0F+G}\right),
\quad D_2H(F,G)=\frac1{RT_0}\log\left(\frac{G}{RT_0F+G}\right).
\end{equation}
\item We \correct{denote by} $\ds \mathbb{H} =
  \left(\begin{smallmatrix}
D_{11}H(F,G)    & \quad D_{12}H(F,G) \\
 D_{12}H(F,G)  & \quad D_{22}H(F,G)
  \end{smallmatrix}\right)$
 the Hessian matrix of $H$. The second order derivatives are:
\begin{align*}
& D_{11}H(F,G)=\frac 2 F -\frac{RT_0}{RT_0F+G}, & D_{12}H(F,G)=-\frac1{RT_0F+G},\\
& D_{21}H(F,G)=D_{12}H(F,G), & D_{22}H(F,G)=\frac{F}{G(RT_0F+G)},
\end{align*}
and we have
\begin{equation}\label{eq-FDG} 
\begin{split}
& FD_{11}H(F,G)+GD_{21}H(F,G) = 1,\\
& FD_{12}H(F,G)+GD_{22}H(F,G) = 0.
\end{split}
\end{equation}

\item The function $(F,G) \mapsto H(F,G)$ is convex.
\end{enumerate}
\end{proposition}
\begin{proof}
Points 1 and 2 are given by direct calculations. For point 3, note
that the determinant of the Hessian matrix $\mathbb{H}$, which is
$\det \mathbb{H}= \frac{1}{G(RT_0F+G)}$ is positive. Moreover, its
trace is positive too, so that the Hessian matrix is positive
definite, and hence the function $H$ is convex.
\end{proof}

Now, we want to use this reduced entropy to define the corresponding
reduced equilibrium. This is done by computing the minimum of the
reduced entropy among all the reduced distributions $(F_1,G_1)$
that have the same moments as $(F,G)$, as it is stated in the following proposition.
%It means that we look for the minimum of $\mathcal{H}(F_1,G_1)$  on the convex set 
%$\mathcal{S}=\left\{(F_1,G_1) \text{ such that} \la F_1\ra_{v,\eps}=\rho ,\quad
%  \la vF_1\ra_{v,\eps}=\rho u, \quad 
%  \la (\frac12|v|^2+\eps) F_1+G_1\ra_{v,\eps} = \rho e \right\}.$ 
%The
%following proposition gives this minimum and other useful differential
%properties of the reduced entropy functional. 
\begin{proposition}[Reduced equilibrium] \label{minentrop}
Let $(F,G)$ be a couple of reduced distributions and $\rho$, $\rho u$, and $\rho e$ its moments as defined by~\eqref{eq-mtsred}. Let ${\cal S}$ be the convex set defined by
\begin{equation*}
    \mathcal{S}=\left\{(F_1,G_1) \text{ such that} \la F_1\ra_{v,\eps}=\rho ,\quad
  \la vF_1\ra_{v,\eps}=\rho u, \quad 
  \la (\frac12|v|^2+\eps) F_1+G_1\ra_{v,\eps} = \rho e \right\}.
\end{equation*}
\begin{enumerate}
\item The minimum of 
  $\mathcal{H}$ on ${\cal S}$
is obtained for the couple $(M_{vib}[F,G], e_{vib}(T)M_{vib}[F,G])$
with
\begin{equation}  \label{eq-MvibFG}
M_{vib}[F,G]=\frac{\rho}{\sqrt{2\pi RT}^3} \exp\left(-\frac{|v - u|^2}{2 RT}\right)  
   \frac{1}{RT} \exp\left(-\frac{ \varepsilon}{RT}\right)
\end{equation}
where $e_{vib}(T)$ is the equilibrium vibrational energy defined by~(\ref{eq-evib}) and $\rho,u,T$ depend on $F$ and $G$ through the definition of the moments.

\item For every $(F_1,G_1)$ in $\mathcal{S}$, we have
\begin{equation*}
\begin{split}
 & D_1H(F_1,G_1)(M_{vib}[F,G]-F_1)+D_2H(F_1,G_1)(e_{vib}(T)M_{vib}[F,G]-G_1) \\
 & \qquad \leq  \ds H(M_{vib}[F,G], e_{vib}(T)M_{vib}[F,G]) -H(F_1,G_1) \leq 0.
\end{split}
\end{equation*}
\end{enumerate}
\end{proposition}
}

\begin{proof}
First, the set $\mathcal{S}$ is
clearly convex, and it is non empty, since it is easy to see that
$(M_{vib}],e_{vib}(T)M_{vib})$ realizes the moments $\rho$, $\rho u$,
and $\rho e$, and hence belongs to $\mathcal{S}$. 
Now, we define the following Lagrangian
\begin{equation*}
\begin{split}
\correct{\mathcal {L}}(F_1,G_1,\alpha,\beta,\gamma)=& \la
  H(F_1,G_1)\ra_{v,\eps}-\alpha(\la F_1\ra_{v,\eps}-\rho) \notag\\
 & -\beta\cdot(\la vF_1\ra_{v,\eps}-\rho u) -\gamma\left(\la
(\frac12|v|^2+\eps)F_1+G_1\ra_{v,\eps} - \rho e\right)
\end{split}
\end{equation*}
for every positive $(F_1,G_1)$, $\alpha\in \mathbb{R}$, $\beta\in
\mathbb{R}^3$, $\gamma\in \mathbb{R}$. The reduced entropy can reach   a
minimum of ${\mathcal S}$  \correct{when $\mathcal {L}$ has its first derivatives equal to zero: it is a minimum if it is unique
}.  \correct{Such a point}, denoted by $(F_1,G_1,\alpha,\beta,\gamma)$
for the moment, is characterised by the fact that the partial
derivatives of $\correct{\mathcal {L}}$ vanish at
$(F_1,G_1,\alpha,\beta,\gamma)$. This gives the following relations \correct{(using the cancellation of the $\mathcal {L} $ derivatives in $F_1,G_1,\alpha,\beta,\gamma$ respectively) }
\begin{eqnarray}
&&D_1H(F_1,G_1)=\alpha+\beta\cdot v+\gamma \frac12|v|^2, \label{Ent1}\\
&&D_2H(F_1,G_1)=\gamma,  \label{Ent2}\\
&&\la F_1\ra_{v,\eps}-\rho=0,\label{Ent3} \\
&&\la v F_1 \ra_{v,\eps}-\rho u=0, \label{Ent4}\\
&&\la (\frac12|v|^2+\eps)F_1+G_1 \ra_{v,\eps}-\rho e=0\label{Ent5},
\end{eqnarray}
where $D_1H$ and $D_2H$ are defined in~(\ref{eq-partialH}). \correct{For instance first relation comes from the fact that the derivative with respect to $F_1$ satisfies for every $\delta F_1$
\begin{equation*}
\begin{split}
& \partial_{{F_1}}{\cal
  L}(F_1,G_1,\alpha,\beta,\gamma)(\delta F_1) \\
& = \langle
(D_1H(F_1,G_1) - (\alpha + \beta \cdot v + \gamma (\frac12
|v|^2 + \varepsilon))) \delta F_1\rangle_{v,\varepsilon},
\end{split}
\end{equation*}
It is true for all $\delta F_1$ leading to the relation \ref{Ent1}.}

\

Now Combining equations (\ref{Ent1}) and (\ref{Ent2}), one gets that there
exists four real numbers $A$, $B$, $C$, $D$ and one vector $E\in
\mathbb{R}^3$, independent of $v$ and $\eps$, such that:
\begin{eqnarray*}
F_1&=&A\exp\left(E\cdot v +B|v|^2+C\eps\right ),\\
G_1&=&D F_1,
\end{eqnarray*}
where $B$ and $C$ are necessarily non positive to ensure the
integrability of $F_1$ and $G_1$.
It is then standard to use equations~(\ref{Ent3})--(\ref{Ent5}) to get
$F_1=M_{vib}(F,G)$ and $G_1=e_{vib}(T)M_{vib}(F,G)$.

Finally point 2 is a direct consequence of the convexity of $H$
and of the minimization property.

\end{proof}

\section{A BGK model with vibrations}
\label{sec:BGK}

With  the local equilibrium $M_{vib}[f]$ defined in~\eqref{maxvib}, it
is easy to derive the following BGK model:
\begin{equation}
 \p_t f + v \cdot \nabla f = \frac{1}{\tau}(M_{vib}[f] - f).
 \label{eq:vib}
\end{equation}
The macroscopic parameters $\rho$, $u$, and $T$ are defined through
the moments $\rho$, $\rho u$ and $\rho e$ of $f$
(see~(\ref{eq-moments})).

Like in the BGK model for monoatomic gases, it will be shown that the
relaxation time of this BGK model is $\tau=\mu/p$, where $p=\rho R T$
is the pressure and $\mu$ the viscosity, that can be temperature
dependent.

Now we have the following properties. 
\begin{prop}
\begin{itemize}
\item Conservation: for BGK model~\eqref{eq:vib} the mass, momentum
  and total energy are conserved:
\begin{equation*}
  \partial_t \la 
  \begin{pmatrix}
    1 \\ v \\ \frac12 |v|^2
  \end{pmatrix}
 f \ra_{v,\eps,i} 
+ \nabla_x \cdot \la 
v  \begin{pmatrix}
    1 \\ v \\ \frac12 |v|^2
  \end{pmatrix}
f \ra_{v,\eps,i}  = 0.
\end{equation*}
\item H-theorem: for the entropy $ {\cal H}(f)=  \la f \log f \ra_{v,\eps,i}$, we have
\begin{equation*}
  \partial_t {\cal H}(f) + \nabla_x \cdot \la vf \log f   \ra_{v,\eps,i} 
= \frac{1}{\tau}\la(M_{vib}[f] - f)\log f \ra_{v,\eps,i} \leq 0.
\end{equation*}
\end{itemize}
\end{prop}
The proof relies on standard arguments (definition of $M_{vib}[f]$ and
convexity of $x\log x$) and is left to the reader.

\subsection{A reduced BGK model with vibrations}

For computational reasons, it is interesting to reduce the complexity
of model~\eqref{eq:vib} by using the usual reduced distribution
technique~\cite{HH_1968}. We define the reduced distributions
\begin{equation*}
  F=\sum_i f(t,x,v,\eps,i), \quad \text{ and } \quad G=\sum_i iRT_0 f(t,x,v,\eps,i),
\end{equation*}
and by summation of~\eqref{eq:vib} on $i$ we get the following closed system of two
reduced equations:
\begin{align}
& \p_t F + v\cdot\nabla_x F  =  \frac1 \tau \left( M_{vib}[F,G] - F \right) \, , \label{eq:eq_marginales_fvib}\\
&  \p_t G + v\cdot\nabla_x G  = \frac1 \tau ( \frac {\delta(T)} 2 RT M_{vib}[F,G] - G ) \, , \label{eq:eq_marginales_gvib}    
\end{align}
where the reduced Maxwellian is
\begin{equation*}
M_{vib}[F,G]=\frac{\rho}{\sqrt{2\pi RT}^3} \exp\left(-\frac{|v - u|^2}{2 RT}\right)   \frac{1}{RT} \exp\left(-\frac{ \varepsilon}{RT}\right),
\end{equation*}
and the macroscopic quantities are defined by
\begin{equation}\label{eq-macro_reduced} 
 \rho = \la F \ra_{v,\eps}, \qquad \rho u= \la vF \ra_{v,\eps},
 \qquad 
 \rho e= \la (\frac12(v-u)^2+\eps)F \ra_{v,\eps}+\la G\ra_{v,\eps},
\end{equation}
and $T$ is still defined by~\eqref{eq-defT} \correct{which implies that $T$ depends both on $F$ and $G$: to avoid
the heavy notation $T[F,G]$, it will still be  denoted by $T$ in the
following.}

Note that this model can easily be reduced once again to eliminate the
rotational energy variable. This gives a reduced system of three
BGK equations, with three distributions.

It is interesting to compare our new model to the work
of~\cite{Mathiaud2018} and~\cite{KKA_2019}: in these recent papers, the
authors also proposed, independently, BGK and ES-BGK models for temperature
dependent \correct{$\delta(T)$}, like in the case of vibrational energy. However,
they are not based on an underlying discrete vibrational energy
partition, and the authors are not able to prove any H-theorem. Only
a local entropy dissipation can be proved. The advantage of our new
approach is that the reduced model, which is continuous in energy too,
inherits the entropy property from the non-reduced model, and hence a
H-theorem, as it is shown below.

\subsection{Properties of the reduced model}

System~(\ref{eq:eq_marginales_fvib}--\ref{eq:eq_marginales_gvib})
naturally satisfies local conservation laws of mass, momentum, and
energy. Moreover, the H-theorem holds with the reduced entropy 
$H(F,G)$ as defined in~(\ref{eq-HFG}). Indeed, we have the 
\begin{proposition}
The reduced BGK
system~(\ref{eq:eq_marginales_fvib}--\ref{eq:eq_marginales_gvib})
satisfies the H-theorem
\begin{equation*}
\partial_t {\cal H}(F,G) +\divx\la v H(F,G)\ra_{v,\eps}\leq 0,
\end{equation*}
where  ${\cal H}(F,G)$ is the reduced entropy defined in~\eqref{eq-HFG}.
\end{proposition}

\begin{proof}
By differentiation we get
\begin{equation*}
\begin{split}
& \partial_t{\cal H}(F,G) +\divx\la v H(F,G)\ra_{v,\eps} \\
&  =  
\la  D_1H(F,G)(\partial_t F + v\nabla_x F)  
  +  D_2H(F,G)(\partial_t G + v\nabla_x G)
\ra_{v,\eps}  \\
&  =  
 \frac1 \tau \la  D_1H(F,G) ( M_{vib}[F,G] - F )
  +  D_2H(F,G)(\frac {\delta(T)} 2 RT M_{vib}[F,G] - G )
\ra_{v,\eps}  \\
&  \leq  0
\end{split}
\end{equation*}
where we have
used~(\ref{eq:eq_marginales_fvib}--\ref{eq:eq_marginales_gvib}) to
replace the transport terms by relaxation ones, and point 5 of
proposition~\ref{minentrop} to obtain the inequality.
\end{proof}

\section{A  Fokker-Planck model with vibrations}
\label{sec:FP}

It is difficult to derive a Fokker-Planck model for the distribution
function $f$ with discrete energy levels. We find it easier to
directly derive
a reduced model, by analogy with the reduced BGK
model~(\ref{eq:eq_marginales_fvib}--\ref{eq:eq_marginales_gvib}) and by
using our previous work~\cite{Mathiaud2017} on
a Fokker-Planck model for polyatomic gases. 
\correct{We remind that the original Fokker-Planck model for monoatomic gas can be derived from
the Boltzmann collision operator under the assumption of small
velocity changes through collisions and additional equilibrium
assumptions (see~\cite{cercignani}). In practice, the agreement
of this model with the Boltzmann equation is observed even when the
gas is far from equilibrium (see~\cite{Grj2011}, for instance).}

%\correct{This approach turned
%out to be very efficient, in particular in the transition regime, since it is shown to be insensitive
%to the number of simulated particles, as opposed to the standard DSMC: comparisons between Fokker-Planck %and DSMC can be found in \cite{gorji2019}.}

\subsection{A  reduced Fokker-Planck model with vibrations}
\label{subsec:FP}

First, we remind the Fokker-Planck model for a diatomic gas (without
vibrations) obtained in~\cite{Mathiaud2017}: 
\begin{equation}\label{eq-FPpoly} 
 \p_t f + v\cdot\nabla_x f  =  D(f),
\end{equation}
where $f=f(t,x,v,\eps)$ and the collision operator is 
\begin{equation*}  
D(f)=\frac{1}{\tau}\left(\divv\bigl((v-u)f+RT\nabla_v
  f\bigr)+2\partial_{\eps}( \eps f+  RT \eps\partial_{\eps} f)\right),
\end{equation*}
where the macroscopic values are
\begin{equation*}
\rho = \la f \ra_{v,\eps}\quad , \quad \rho u= \la fv \ra_{v,\eps}\quad ,\quad 
\rho e= \la f\left(\frac12(v-u)^2+\eps\right)\ra_{v,\eps}=\frac52\rho RT .  
\end{equation*}
The internal energy $\eps$ can be eliminated by the reduction technique
(integration w.r.t $d\eps$ and $\eps d\eps$) to get
% \begin{align*}
% & \p_t  \la f \ra_{v,\eps} + v\cdot\nabla_x  \la f \ra_{v,\eps}  =  D_1( \la f \ra_{v,\eps}, \la f\eps \ra_{v,\eps}),  \\
% & \p_t  \la f\eps \ra_{v,\eps} + v\cdot\nabla_x  \la f\eps \ra_{v,\eps}  =  D_2( \la f \ra_{v,\eps}, \la f\eps \ra_{v,\eps}) ,\\   
% \end{align*}
\begin{align*}
& \p_t  {\cal F} + v\cdot\nabla_x  {\cal F}  =  D_1( {\cal F}, {\cal G}),  \\
& \p_t  {\cal G} + v\cdot\nabla_x  {\cal G}  =  D_2( {\cal F}, {\cal G}) , 
\end{align*}
with the collision operators
\begin{align*}
&D_{\cal F}( {\cal F}, {\cal G})=\frac{1}{\tau}\divv\bigl((v-u) {\cal F}+RT\nabla_v  {\cal F}\bigr),\\
&D_{\cal G}( {\cal F}, {\cal G})=\frac{1}{\tau}\divv\bigl((v-u) {\cal G}+RT\nabla_v  {\cal G}\bigr)
 +\frac2{\tau}\left(RT {\cal F}- {\cal G}\right).
\end{align*}
Note that the two velocity drift-diffusion terms in the two previous equations
have exactly the same structure as the one in the non-reduced
model~(\ref{eq-FPpoly}). However, it is interesting to note that the
energy drift-diffusion term of~(\ref{eq-FPpoly}) gives, after
reduction, a relaxation operator in the ${\cal G}$ equation.   \correct{Moreover by reducing  the model we lose some moments of initial distribution functions (higher moments in internal energy notably) but we are still able to capture energies and fluxes which are generally the main quantities of interest}.

By analogy, now we propose the following reduced Fokker-Planck model for a
diatomic gas with vibrations. Note that now, the model is still with
variables $x$, $v$, and $\eps$: only the discrete energy levels $i$
are eliminated. 
This model is 
\begin{align}
&  \p_t F + v\cdot\nabla_x F  =  D_F(F,G),\label{eq-FPF}  \\
&  \p_t G + v\cdot\nabla_x G  =  D_G(F,G) \label{eq-FPG},     % (t,x,v)
\end{align}
with
\begin{equation}\label{eq-DFDG}   
\begin{split}
& D_F(F,G)=\frac{1}{\tau}\left(\divv\bigl((v-u)F+RT\nabla_v
  F\bigr)+2\partial_{\eps}( \eps F+  RT \eps\partial_{\eps} F)\right),\\
& D_G(F,G)= \frac{1}{\tau}\left(\divv\bigl((v-u)G+RT\nabla_v
  G\bigr)+2\partial_{\eps}(\eps G+  RT \eps\partial_{\eps} G)\right)
 +\frac2{\tau}\left(e_{vib}(T)F-G\right),
\end{split}
\end{equation}
where the macroscopic values are defined as
in~(\ref{eq-macro_reduced}) and ~\eqref{eq-defT}.
% \begin{eqnarray}  
% \rho &=& \la F \ra_{v,\eps},\\
% \rho u&=& \la Fv \ra_{v,\eps},\\
% \rho e&=& \la F\left(\frac12(v-u)^2+\eps\right) \ra_{v,\eps}+\la G\ra_{v,\eps},\\
% T&=&T^{-1}(e).
% \end{eqnarray}
\correct{Again, note that the temperature $T$ depends on $F$ and $G$.}

\correct{Note that we do not derive
this reduced Fokker-Planck model directly from a model with discrete
vibrational energy as for the BGK model,
since we are not able so far to define a discrete diffusion
operator. As mentioned above, this model is obtained by analogy with the
Fokker-Plank model proposed for polyatomic gases. Its derivation from
reduction of a discrete in energy Fokker-Plank model will be studied
in a future work.}

\subsection{Properties of the reduced model}

Using direct calculations and dissipation properties as in \cite{Mathiaud2017} we can prove the following proposition.
\begin{proposition} \label{prop:FP}
The collision operator conserves the mass, momentum, and energy:
\begin{equation*}
  \la(1,v)D_F(F,G)\ra_{v,\eps} = 0 \quad \text{ and } \quad
\la(\frac12 |v|^2+ \eps)D_F(F,G) + D_G(F,G)\ra_{v,\eps} = 0,
\end{equation*}
the reduced entropy $\mathcal{H}(F,G)$ satisfies the H-theorem:
\begin{equation*}
\partial_t \mathcal{H}(F,G) +\divx\la v H(F,G)\ra_{v,\eps}=
  \mathcal{D}(F,G)\leq 0,
\end{equation*}
and we have the equilibrium property
\begin{equation*}
  (D_F(F,G) = 0 \text{ and } D_G(F,G) = 0) \Leftrightarrow
(F = M_{vib}[F,G]  \text{ and }  G = e_{vib(T)}M_{vib}[F,G] ).
\end{equation*}

\end{proposition}
\begin{proof}
The conservation property is the consequence of direct integration
of~\eqref{eq-DFDG}. The equilibrium property can be proved as
follows. First, note that the Maxwellian $M_{vib}[F,G]$ can be written
as
\begin{equation*}
  M_{vib}[F,G] = \frac{\rho}{(2\pi)^{3/2} (RT)^{5/2}} \exp\left(
-\frac12 
\begin{pmatrix}
  v-u \\ 2\eps
\end{pmatrix}^T
\Omega^{-1} 
\begin{pmatrix}
  v-u \\ 2\eps
\end{pmatrix}\right),
\end{equation*}
with $\Omega =
\begin{pmatrix}
  RT & 0 \\ 0 & 2\eps RT
\end{pmatrix}
$. To shorten the notations, $M_{vib}[F,G]$ will be simply
denoted by $M_{vib}$ below, and $e_{vib}(T)$ will be simply denoted by
$e_{vib}$ as well. Then the collision operators can be written in the compact form
\begin{equation*}
\begin{split}
& D_F(F,G)=\frac{1}{\tau}\nabla_{v,\eps} \cdot
\left( \Omega M_{vib} \nabla_{v,\eps} \frac{F}{M_{vib} }  \right),\\
& D_G(F,G)= \frac{1}{\tau}\nabla_{v,\eps} \cdot
\left( \Omega M_{vib} \nabla_{v,\eps} \frac{G}{M_{vib} }  \right)
 +\frac2{\tau}\left(e_{vib}F-G\right).
\end{split}
\end{equation*}
Then an integration by part gives the following identity for
$D_F(F,G)$: 
\begin{equation*}
  \la D_F(F,G) \frac{F}{M_{vib}}\ra_{v,\eps} =  
- \frac{1}{\tau}\la \left(\nabla_{v,\eps} \frac{F}{M_{vib} } \right)^T\Omega M_{vib} \nabla_{v,\eps} \frac{F}{M_{vib} } \ra_{v,\eps}.
\end{equation*}
Consequently, if $D_F(F,G)=0$, 
since the integrand in the previous relation is a definite positive
form, the gradient is necessarily zero, and hence
$F=M_{vib}$. 
For the equilibrium property of $G$, the proof is a bit more
complicated. First, we have 
\begin{equation*}
  \la D_G(F,G) \frac{G}{e_{vib}M_{vib}}\ra_{v,\eps} =  
- \frac{1}{\tau e_{vib}}\la \left(\nabla_{v,\eps}
\frac{G}{M_{vib} }\right)^T \Omega M_{vib} \nabla_{v,\eps}
\frac{G}{M_{vib} } \ra_{v,\eps} 
+\la \frac2{\tau}\left(e_{vib}F-G\right)\frac{G}{e_{vib}M_{vib}}\ra_{v,\eps}.
\end{equation*}
Consequently, if $D_G(F,G)=0$, and since $F=M_{vib}$, we have
\begin{equation*}
\begin{split}
  \frac{1}{e_{vib}}\la \left(\nabla_{v,\eps}
\frac{G}{M_{vib} }\right)^T \Omega M_{vib} \nabla_{v,\eps}
\frac{G}{M_{vib} } \ra_{v,\eps}  
& = \frac{2}{\tau}
\la\left(e_{vib}M_{vib}-G\right)\frac{G}{e_{vib}M_{vib}}\ra_{v,\eps}
\\
&= - \frac{2}{\tau}
\la\left(e_{vib}M_{vib}-G\right)^2
      \frac{1}{e_{vib}M_{vib}}
\ra_{v,\eps}
+  \frac{2}{\tau}\la  e_{vib}M_{vib}-G\ra_{v,\eps} \\
& \leq \frac{2}{\tau}\la  e_{vib}M_{vib}-G\ra_{v,\eps} 
 = \frac{2}{\tau} (\rho e_{vib}-\la G\ra_{v,\eps}) = 0,
\end{split}
\end{equation*}
which comes from~\eqref{eq-mtsred} and $F=M_{vib}$. Therefore, we
obtain
\begin{equation*}
   \frac{1}{e_{vib}}\la \left(\nabla_{v,\eps}
\frac{G}{M_{vib} }\right)^T \Omega M_{vib} \nabla_{v,\eps}
\frac{G}{M_{vib} } \ra_{v,\eps}  \leq 0,
\end{equation*}
and again this gives $G = e_{vib}M_{vib}$,
which concludes the proof of the equilibrium property.

The proof of the H-theorem is much longer. First, by differentiation one gets that the quantity $\mathcal{D}(F,G)=\partial_t \mathcal{H}(F,G) +\divx\la vH(F,G)\ra_{v,\eps}$  satisfies:
\begin{eqnarray}
\mathcal{D}(F,G)&=&\la D_1H(F,G)(\partial_tF +v\cdot\nabla_x F )
                      +D_2H(F,G)(\partial_tG +v\cdot\nabla_x G)\ra_{v,\eps}\notag\\
&=& \la D_1H(F,G)D_F(F,G)+D_2H(F,G)D_G(F,G)\ra_{v,\eps} ,
\end{eqnarray}
from~(\ref{eq:eq_marginales_fvib}--\ref{eq:eq_marginales_gvib}).
Then the  proof is based on the convexity of $H(F,G)$: while for the BGK model we
only used the  first derivatives of $H$, we now use the
positive-definiteness of the Hessian matrix of $H$. To do so we
integrate by parts $\mathcal{D}(F,G)$ and multiply by $\tau$ so that:
\begin{eqnarray*}
\tau\mathcal{D}(F,G)&=&-\sum_{i=1}^3\la\partial_{v_i}(F)D_{11}H(F,G)\left(F(v_i-u_i)+RT\partial_{v_i} F\right)\ra_{v,\eps}\\
&&-2\la\partial_{\eps}(F)D_{11}H(F,G)\left(F\eps+RT\eps\partial_{\eps} F\right)\ra_{v,\eps}\\
&&-\sum_{i=1}^3\la\partial_{v_i}(G)D_{21}H(F,G)\left(F(v_i-u_i)+RT\partial_{v_i} F\right)\ra_{v,\eps}\\
&&-2\la\partial_{\eps}(G)D_{21}H(F,G)\left(F\eps+RT\eps\partial_{\eps} F\right)\ra_{v,\eps}\\
&&-\sum_{i=1}^3\la\partial_{v_i}(F)D_{12}H(F,G)\left(G(v_i-u_i)+RT\partial_{v_i} G\right)\ra_{v,\eps}\\
&&-2\la\partial_{\eps}(F)D_{12}H(F,G)\left(G\eps+RT\eps\partial_{\eps} G\right)\ra_{v,\eps}\\
&&-\sum_{i=1}^3\la\partial_{v_i}(G)D_{22}H(F,G)\left(G(v_i-u_i)+RT\partial_{v_i} G\right)\ra_{v,\eps}\\
&&-2\la\partial_{\eps}(G)D_{22}H(F,G)\left(G\eps+RT\eps\partial_{\eps} G\right)\ra_{v,\eps}\\
&&+2\la(e_{vib}(T)F-G)\frac 1{RT_0}\log\left(\frac{G}{RT_0F+G}\right)\ra_{v,\eps}
\end{eqnarray*}
To use the positive definiteness of the Hessian matrix $\mathbb{H}$ of
$H$, we introduce the following vectors:
\begin{align*}
&   V_i=(F(v_i-u_i)+RT\partial_{v_i} F,G(v_i-u_i)+RT\partial_{v_i} G) \\   
& E=(F\eps+RT\eps\partial_{\eps} F,G\eps+RT\eps\partial_{\eps} G),
\end{align*}
and we decompose the partial derivatives of $F$ and $G$ in factor of $D_{11}F$,
$D_{22}F$, $D_{12}F$ as follows:
\begin{align*}
& (\partial_{v_i}(F),\partial_{v_i}(G)) = \frac{1}{RT}V_i
  -(F\frac{v_i-u_i}{RT},G\frac{v_i-u_i}{RT}) \\
& (\partial_{\eps}(F), \partial_{\eps}(G)) = \frac{1}{\eps}E - (F\frac{1}{RT},G\frac{1}{RT}).
\end{align*}
This gives
\begin{eqnarray*}
\tau\mathcal{D}(F,G)&=&\sum_{i=1}^3\la\left(F\frac{v_i-u_i}{RT}\right)D_{11}H(F,G)\left(F(v_i-u_i)+RT\partial_{v_i} F\right)\ra_{v,\eps}\\
&&+2\la \left(F\frac{1}{RT}\right)D_{11}H(F,G)\left(F\eps+RT\eps\partial_{\eps} F\right)\ra_{v,\eps}\\
&&+\sum_{i=1}^3\la\left(G\frac{v_i-u_i}{RT}\right)D_{21}H(F,G)\left(F(v_i-u_i)+RT\partial_{v_i} F\right)\ra_{v,\eps}\\
&&+2\la\left(G\frac{1}{RT}\right)D_{21}H(F,G)\left(F\eps+RT\eps\partial_{\eps} F\right)\ra_{v,\eps}\\
&&+\sum_{i=1}^3\la\left(F\frac{v_i-u_i}{RT}\right) D_{12}H(F,G)\left(G(v_i-u_i)+RT\partial_{v_i} G\right)\ra_{v,\eps}\\
&&+2\la\left(f\frac{1}{RT}\right)D_{12}H(F,G)\left(g\eps+RT\eps\partial_{\eps} G\right)\ra_{v,\eps}\\
&&+\sum_{i=1}^3\la\left(G\frac{v_i-u_i}{RT}\right)D_{22}H(F,G)\left(G(v_i-u_i)+RT\partial_{v_i} G\right)\ra_{v,\eps}\\
&&+2\la\left(G\frac{1}{RT}\right)D_{22}H(F,G)\left(G\eps+RT\eps\partial_{\eps} G\right)\ra_{v,\eps}\\
&&- \sum_{i=1}^3\la V_i^T\mathbb{H}V_i\ra_{v,\eps}-2\la E^T\mathbb{H}E\ra_{v,\eps}\\
&&+2\la(e_{vib}(T)F-G)\frac 1{RT_0}\log\left(\frac{G}{RT_0F+G}\right)\ra_{v,\eps}
\end{eqnarray*}

Now this expression can be considerably simplified by using
property~(\ref{eq-FDG}), and we get
\begin{eqnarray*}
\tau\mathcal{D}(F,G)&=&\sum_{i=1}^3\la\left(\frac{v_i-u_i}{RT}\right)\left(F(v_i-u_i)+RT\partial_{v_i} F\right)\ra_{v,\eps}\\
&&+2\la \frac{1}{RT}\left(F\eps+RT\eps\partial_{\eps} F\right)\ra_{v,\eps}\\
&&-\sum_{i=1}^3 V_i^t\mathbb{H}V_i-E^t\mathbb{H}E\\
&&-2\la(e_{vib}(T)F-G)\frac 1{RT_0}\log\left(\frac{G}{RT_0F+G}\right)\ra_{v,\eps}.
\end{eqnarray*}
Then the first two terms are simplified by using an integration by
parts and relations~(\ref{eq-mtsred}) and~(\ref{eq-defT}) to get
\begin{eqnarray*}
\tau\mathcal{D}(F,G)&=& \frac{2}{RT}(\rho e_{vib}(T) - \la G\ra_{v,\eps}) \\
&&-\sum_{i=1}^3 V_i^t\mathbb{H}V_i-2E^t\mathbb{H}E\\
&&+2\la(e_{vib}(T)F-G)\frac 1{RT_0}\log\left(\frac{G}{RT_0F+G}\right)\ra_{v,\eps}.
\end{eqnarray*}
The terms with the Hessian are clearly negative, since $\mathbb{H}$ is
positive definite. Then we have
\begin{eqnarray*}
\tau\mathcal{D}(F,G)&\leq& \frac{2}{RT}(\rho e_{vib}(T) - \la G\ra_{v,\eps}) \\
&&+2\la(e_{vib}(T)F-G)\frac 1{RT_0}\log\left(\frac{G}{RT_0F+G}\right)\ra_{v,\eps}.
\end{eqnarray*}
Note that from~\eqref{eq-mtsred} the first term can be written as
\begin{equation*}
  \frac{2}{RT}(\rho e_{vib}(T) - \la G\ra_{v,\eps})  =
  \frac{2}{RT}\la e_{vib}(T) F-G\ra_{v,\eps}, 
\end{equation*}
and can be factorized with the second term to find
\begin{equation*}
\tau\mathcal{D}(F,G)\leq
2\la(e_{vib}(T)F-G)\left(\frac
  1{RT_0}\log\left(\frac{G}{RT_0F+G}\right) + \frac{1}{RT}\right)\ra_{v,\eps}.
\end{equation*}
We can now prove that the integrand of the right-hand side is non-positive. Indeed,
assume for instance that the first factor is non-positive, that is to
say $e_{vib}(T)F - G\leq 0$. By using
$e_{vib}(T)=\frac{RT_0}{e^{T_0/T}-1}$ (see definition~(\ref{eq-evib})), it is now very easy to prove the following relations
\begin{equation*}
 e_{vib}(T)F - G\leq 0 \Leftrightarrow   \frac{1}{T_0}\log\left(\frac{G}{RT_0F+G}\right) \geq -\frac{R}{T}
\end{equation*}
that is to say the second factor of the integrand is non-negative.

Consequently, we have proved $\tau\mathcal{D}(F,G)\leq 0$, 
which concludes the proof. 
\end{proof}

\section{Hydrodynamic limits for reduced models}
\label{sec:CE}

\correct{With a convenient nondimensionalization}, the relaxation time $\tau$ of the reduced BGK
model~\eqref{eq:eq_marginales_fvib}--\eqref{eq:eq_marginales_gvib} and
the Fokker-Planck model~(\ref{eq-FPF})-(\ref{eq-FPG}) is replaced by $\Kn \tau$, 
where $\Kn$ is the Knudsen number, which can be defined as a ratio
between the mean free path and a macroscopic length scale.
It is then possible to look for macroscopic models derived from BGK and
Fokker-Planck reduced models, in the asymptotic limit of small Knudsen
numbers.  

For convenience, these models are re-written below in non-dimensional
form. The BGK model is
\begin{align}
& \p_t F + v\cdot\nabla_x F  =  \frac{1}{\Kn\tau} \left( M_{vib}[F,G] - F \right) \, , \label{eq:BGKf}\\
&  \p_t G + v\cdot\nabla_x G  = \frac{1}{\Kn\tau}  ( \frac {\delta(T)} 2 T M_{vib}[F,G] - G ) \, , \label{eq:BGKg}    
\end{align}
where $M_{vib}[F,G]$ can be defined by~(\ref{eq-MvibFG}) with
$R=1$. Similarly, the relations~\eqref{eq-etr}--\eqref{eq-defT}
between the translational, internal, and total energies and the
temperature, have to be read with $R=1$ in  non-dimensional variables. \correct{We remind that $T$ is still a function of $F$ and $G$}. The Fokker-Planck model is
\begin{align}
&  \p_t F + v\cdot\nabla_x F  =  D_F(F,G),\label{eq-FPf}  \\
&  \p_t G + v\cdot\nabla_x G  =  D_G(F,G) \label{eq-FPg},     % (t,x,v)
\end{align}
with
\begin{equation}\label{eq-DFDGadim}   
\begin{split}
& D_F(F,G)=\frac{1}{\Kn\tau}\left(\divv\bigl((v-u)F+T\nabla_v
  F\bigr)+2\partial_{\eps}( \eps F+  T \eps\partial_{\eps} F)\right),\\
& D_G(F,G)= \frac{1}{\Kn\tau}\left(\divv\bigl((v-u)G+T\nabla_v
  G\bigr)+2\partial_{\eps}(\eps G+  T \eps\partial_{\eps} G)\right)
 +\frac2{\Kn\tau}\left(e_{vib}(T)F-G\right).
\end{split}
\end{equation}

\subsection{Euler limit}

In this section, we prove the following proposition:
\begin{proposition}
  The mass, momentum, and energy densities $(\rho, \rho u, E = \frac12 \rho
u^2 + \rho e)$ of the solutions of the
  reduced BGK and the Fokker-Planck models satisfy the equations
\begin{equation}\label{eq-euler} 
\begin{split}
& \dt \rho + \nabla_x \cdot \rho u = 0, \\
& \dt \rho u + \nabla_x\cdot  (\rho u\otimes u) + \nabla p = O(\Kn), \\
& \dt E + \nabla_x \cdot (E+p)u =O(\Kn),
\end{split}
\end{equation}
which are the Euler equations, up to $O(\Kn)$. The non-conservative
form of these equations is
\begin{equation}\label{eq-euler_non_cons} 
\begin{split}  & \dt \rho + \nabla_x \cdot \rho u = 0, \\
& \rho (\dt u + (u \cdot \nabla_x) u) + \nabla p = O(\Kn), \\
&  \dt T + u\cdot \nabla_x T+\frac{T}{c_v(T)}\divx u =O(\Kn),
\end{split}
\end{equation}
where $c_v(T)=\frac{d}{dT}e(T)$ is the \correct{specific} heat capacity  at constant volume.
\end{proposition}

\begin{proof}
The reduced BGK
  model~\eqref{eq:eq_marginales_fvib}--\eqref{eq:eq_marginales_gvib}
is multiplied by $1$, $v$, and $\frac12 |v|^2 + \eps$ and integrated
with respect to $v$ and $\eps$, which gives the following conservation
laws:
\begin{equation}\label{eq-conslaws} 
\begin{split}
  & \dt \rho + \nabla_x \cdot \rho u = 0, \\
& \dt \rho u + \nabla_x\cdot  (\rho u\otimes u) + \nabla_x \sigma(F) = 0, \\
& \dt E + \nabla_x \cdot Eu + \nabla_x \cdot \sigma(F)   u + \nabla_x \cdot q(F,G)=0,
\end{split}
\end{equation}
where $\sigma(F) = \la (v-u)\otimes (v-u) F \ra_{v,\eps}$ is the stress
tensor, and $q(F,G) =  \la  (v-u)(\frac12 |v-u|^2 + \eps) F \ra_{v,\eps}
+\la  (v-u) G \ra_{v,\eps} $ is the heat flux.

When $\Kn$ is very small, if all the time and space derivatives of $F$
and $G$ are $O(1)$ with respect to $\Kn$,
then~(\ref{eq:BGKf})--(\ref{eq:BGKg})  imply
$F = M_{vib}[F,G] + O(\Kn)$ and $G = e_{vib}(T)
M_{vib}[F,G] +O(\Kn)$. Then it is easy to find that $\sigma(F) =
\sigma(M_{vib}[F,G]) + O(\Kn) = p I + O(\Kn)$ ,
where $I$ is the unit tensor, and $q(F,G) =
q(M_{vib}[F,G],e_{vib}(T)M_{vib}[F,G])+ O(\Kn) = O(\Kn)$, \correct{since the heat flux is zero at equilibrium},
which gives the Euler equations~(\ref{eq-euler_non_cons}).
The same analysis can be applied for the reduced Fokker-Planck model~(\ref{eq-FPf})--(\ref{eq-DFDGadim}).

Finally, the non conservative form is readily obtained from the
conservative form. Note another formulation of the energy equation that 
will be useful below: 
\begin{eqnarray}
 \dt e_{vib}(T)+ u\cdot \nabla_x
 e_{vib}(T)+\frac{Te_{vib}'(T)}{c_v(T)}\divx u =O(\Kn),
\end{eqnarray}
where $e_{vib}'(T)=\frac{d}{dT} e_{vib}(T)$.
\end{proof}

% As a consequence the internal energy satisfies:
% \begin{eqnarray}
%  \dt e + u\cdot \nabla_x e+\frac{p}{\rho}\divx(u) =0+O(\eps),
% \end{eqnarray}
% If one notes $\ds c_v(T):=\frac{d}{dT}(e(T))$, one finally gets:
% \begin{eqnarray}
%  \dt T + u\cdot \nabla_x T+\frac{RT}{c_v(T)}\divx(u) =0+O(\eps).
% \end{eqnarray}

\subsection{Compressible Navier-Stokes limit}
In this section, we shall prove the following proposition:
\begin{proposition}
  The moments of the solution of the BGK and Fokker-Planck kinetic
  models~\eqref{eq:eq_marginales_fvib}-\eqref{eq:eq_marginales_gvib}
  and~\eqref{eq-FPF}-\eqref{eq-FPG} satisfy,
  up to $O(\Kn^2)$, the Navier-Stokes equations
\begin{equation}\label{eq-ns}
\begin{split}
& \dt \rho + \nabla \cdot \rho u = 0, \\
& \dt \rho u + \nabla\cdot  (\rho u\otimes u) + \nabla p = -\nabla
\cdot \sigma, \\
& \dt E + \nabla \cdot (E+p)u = -\nabla\cdot q - \nabla\cdot(\sigma u),
\end{split}
\end{equation}
where the shear stress tensor and the heat flux are given by
\begin{equation}  \label{eq-fluxes_ns}
\sigma = -\mu \bigl(\nabla u + (\nabla u)^T -\alpha\nabla\cdot
u\bigr), \quad \text{and} \quad  q=-\kappa \nabla \cdot T,
\end{equation}
and where the following values of the viscosity and heat
transfer coefficients (in dimensional variables) are
\begin{equation}  \label{eq-coef}
\begin{split}
& \mu = \tau p, \quad \text{and} \quad
\kappa = \mu c_p(T) \quad \text{for BGK}, \\
& \mu = \frac{1}{2}\tau p, \quad \text{and} \quad
\kappa = \frac{2}{3}\mu c_p(T) \quad \text{for Fokker-Planck},
\end{split}
\end{equation}
 while the volumic viscosity coefficient is $
 \alpha=\frac{c_p(T)}{c_v(T)}-1 $ for both models, and  $c_p(T)=\frac{d}{dT}(e(T)+p/\rho) = c_v(T) +R$ is the \correct{specific} heat capacity at constant pressure.
Moreover,
the corresponding Prandtl number is
\begin{equation}  \label{eq-PrM}
\Pr = \frac{\mu c_p(T)}{\kappa}=1 \quad \text{for BGK}, \quad \text{and} \quad \frac{3}{2}
\quad \text{for Fokker-Planck}.
\end{equation}

\end{proposition}

\correct{Note that both models do not provide a correct Prandtl number, which can lead to errors for the computation of fluxes in numerical simulations. This is a usual problem with single parameter models like BGK or Fokker-Planck: this can be corrected by a modification of the models like with the ES-BGK or ES-FP approaches, as it has been done for polyatomic gases (see~\cite{esbgk_poly,Mathiaud2017} for instance)}.

\subsubsection{Proof for the BGK model}
\label{sec:proof-bgk-model}
The usual Chapman-Enskog method is applied as follows. We decompose
$F$ and $G$ as $F = M_{vib}[F,G] + \Kn F_1$ and $G = e_{vib}(T)
M_{vib}[F,G] +\Kn G_1$, which gives
\begin{equation*}
  \sigma(F) = p I + \Kn \sigma(F_1), \qquad \text{and} \qquad
q(F,G) = \Kn q (F_1,G_1).
\end{equation*}
Then we have to approximate $\sigma(F_1)$ and $q (F_1,G_1)$ up to
$O(\Kn)$. This is done by using the previous expansions
and~\eqref{eq:eq_marginales_fvib} and~\eqref{eq:eq_marginales_gvib} to
get
\begin{equation*}
\begin{split}
& F_1 = -\tau (\p_t M_{vib}[F,G] + v\cdot\nabla_x M_{vib}[F,G]) +O(\Kn), \\
& G_1 = -\tau (\p_t e_{vib}(T)M_{vib}[F,G] + v\cdot\nabla_x e_{vib}(T)M_{vib}[F,G]) + O(\Kn). \\
\end{split}
\end{equation*}
This gives the following approximations
\begin{equation}
  \sigma(F_1) = -\tau\la
(v-u)\otimes (v-u) (\p_t M_{vib}[F,G] + v\cdot\nabla_x M_{vib}[F,G])
  \ra_{v,\eps}   +O(\Kn),\label{eq-sigmaF1}
\end{equation}  
and
\begin{equation}
\begin{split}
 q(F_1,G_1) =&  -\tau  \la
(v-u)(\frac12 |v-u|^2 + \eps)(\p_t M_{vib}[F,G] + v\cdot\nabla_x M_{vib}[F,G])
 \ra_{v,\eps} \\
&  -\tau\la  
(v-u) (\p_t e_{vib}(T)M_{vib}[F,G] + v\cdot\nabla_x e_{vib}(T)M_{vib}[F,G] )
\ra_{v,\eps}
+ O(\Kn).\label{eq-qF1G1}
\end{split}
\end{equation}

Now it is standard to write $\dt M_{vib}[F,G]$ and $\nabla_x M_{vib}[F,G]$ as
functions of derivatives of $\rho$, $u$, and $T$, and then to use
Euler equations~\eqref{eq-euler} to write time derivatives as
functions of the space derivatives only. After some algebra, we get
\begin{equation}\label{eq-TMvib} 
 \dt\left(M_{vib}(F,G)\right) + v \cdot \nabla_x \left(M_{vib}(F,G)\right) = \frac{\rho}{T^{\frac{5}{2}}}M_0(V)e^{-J}\left(A \cdot \frac{\nabla
    T}{\sqrt{T}} + B : \nabla u \right)+ O(\Kn),
\end{equation}
where 
\begin{align*}
& V=\frac{v-u}{\sqrt{T}}, \qquad J=\frac{\eps} {T} , \qquad M_0(V) =
\frac{1}{(2\pi)^{\frac32}}\exp(-\frac{|V|^2}{2})\\
& A = \left(\frac{|V|^2}{2}+J-\frac{7}{2}\right)V, 
\qquad B = V\otimes V - \left(\frac{1}{c_v}\left(\frac12|V|^2+J\right)+\frac{e_{vib}'(T)}{c_{v}(T)}\right)I.
\end{align*}
Then we introduce~\eqref{eq-TMvib} into~\eqref{eq-sigmaF1} to get
\begin{equation*}
  \sigma_{ij}(F_1) = -\tau \rho T \la
  V_iV_jB_{kl}M_0e^{-J}\ra_{V,J}\partial_{x_l}u_k + O(\Kn),
\end{equation*}
where we have used the change of variables $(v,\eps)\mapsto (V,J)$ in the
integral (the term with $A$ vanishes due to the parity of
$M_0$). Then standard Gaussian integrals (see~appendix~\ref{app:CE}) give
\begin{equation*}
  \sigma(F_1) = -\mu \left(\nabla u + (\nabla u)^T
      -\alpha \nabla \cdot u \, I\right) + O(\Kn), 
\end{equation*}
with $\mu = \tau \rho T$ and $\alpha= \frac{c_p}{c_v}-1$, 
which is the announced result, in a non-dimensional form.

For the heat flux, we use the same technique. First
for $e_{vib}(T)M_{vib}[F,G]$ we obtain
\begin{equation}\label{eq-TevibMvib} 
 \dt\left(e_{vib}M_{vib}(F,G)\right) + v \cdot \nabla_x \left(e_{vib}M_{vib}(F,G)\right) = \frac{\rho}{T^{\frac{3}{2}}}M_0(V)\left(\tilde A \cdot \frac{\nabla
    T}{\sqrt{T}} + \tilde B : \nabla u\right) + O(\Kn),
\end{equation}
where
\begin{align*}
&   \tilde A = \left(\frac{|V|^2}{2}+J-\frac{7}{2}+\frac{Te_{vib}'(T)}{e_{vib}}\right)V ,\\
&  \tilde B = V\otimes V - \left(\frac{1}{c_v}\left(\frac12|V|^2+J\right)+\frac{e_{vib}'(T)}{c_{v}(T)}+\frac{Te_{vib}'(T)}{c_v(T)e_{vib}}\right)I.
\end{align*}
Then $q(F_1,G_1)$ as
given in~\eqref{eq-qF1G1} can be reduced to
\begin{equation*}
\begin{split}
q_i(F_1,G_1) & = 
-\tau \rho T\left( 
\la  \frac{1}{2}|V|^2V_iA_jM_0e^{-J}\ra_{V,J}
+ \la  V_i J A_jM_0e^{-J}\ra_{V,J}
\right)\partial_{x_j}T \\
& \qquad 
- \tau \rho 
 \la  V_i\tilde{A}_jM_0e^{-J}\ra_{V,J}\partial_{x_j}T.
\end{split}
\end{equation*}
Using again Gaussian integrals , we get
\begin{equation*}
q(F_1,G_1)=-\kappa\nabla_x T,  
\end{equation*}
where $\kappa = \mu c_p(T)$ with  $c_p(T)= \frac{d}{dT}(e(T) + \frac{p}{\rho}) = \frac{7}{2} +
e_{vib}'(T)= 1 + c_v(T)$ in a non-dimensional form.

\subsubsection{Proof for the Fokker-Planck model}

Here, we rather use the decomposition $F = M_{vib}(1 + \Kn F_1)$ and $G = e_{vib}
M_{vib} (1+\Kn G_1)$, which gives 
\begin{equation*}
  \sigma(F) = p I + \Kn \sigma(M_{vib}F_1) \quad \text{and} \quad
  q(F,G) = \Kn q (M_{vib}F_1,e_{vib}M_{vib}G_1),
\end{equation*}
in which, for clarity, the dependence of $M_{vib}$ on $F$ and $G$ has been omitted,
and the dependence of $e_{vib}$ on $T$ as well.
Finding $F_1$ and $G_1$ is less simple than for the BGK
model: however, the computations are very close to what is done in the
standard monatomic Fokker-Planck model (see~\cite{Mathiaud2016} for instance), so
that we only give the main steps here (see appendix~\ref{app:CE} for
details).

First, the decomposition is injected into~(\ref{eq-DFDGadim}) to get
\begin{align*}
&  D_F(F,G) = \frac{1}{\tau}M_{vib}L_F(F_1) + O(\Kn) , \\
&  D_G(F,G) = \frac{1}{\tau}e_{vib}M_{vib}L_G(F_1,G_1) + O(\Kn) , 
\end{align*}
where $L_F$ and $L_G$ are linear operators defined by
\begin{equation} \label{eq-LFLG} 
\begin{split}
& L_F(F_1) = \frac{1}{M_{vib}}\Bigl(\divv (TM_{vib}\nabla_v F_1)+\partial_{\eps} \left(2 T {\eps}M_{vib}\partial_{\eps} F_1\right)\Bigr),\\
& L_G(F_1,G_1) = \frac{1}{e_{vib}M_{vib}}
 \Bigl( \divv (Te_{vib}M_{vib}\nabla_v G_1)
+2\partial_{\eps} \left( T {\eps}e_{vib}M_{vib}\partial_{\eps} G_1\right)+2(F_1-G_1)
 \Bigr).
\end{split}
\end{equation}

Then the Fokker-Planck equations~(\ref{eq-FPf})-(\ref{eq-FPg}) suggest to look
for an approximation of $F_1$ and $G_1$ up to $O(\Kn)$ as solutions of 
\begin{align*}
&  \p_t M_{vib} + v\cdot\nabla_x M_{vib} =
  \frac{1}{\tau}M_{vib}(F,G)L_F(F_1) \\
& \p_t e_{vib}M_{vib} + v\cdot\nabla_x e_{vib}M_{vib}
 =   \frac{1}{\tau}e_{vib}M_{vib}(F,G)L_G(F_1,G_1).
\end{align*}
By using~(\ref{eq-TMvib})-(\ref{eq-TevibMvib}), these relations are
equivalent, up to
another $O(\Kn)$ approximation, to 
\begin{equation}\label{eq-F1G1FP} 
L_F(F_1) = \tau \left(A \cdot \frac{\nabla T}{\sqrt{T}} 
+ B:\nabla u\right), 
\quad \text{ and } \quad
 L_G(F_1,G_1) = \tau \left(\tilde{A} \cdot \frac{\nabla T}{\sqrt{T}} 
+ \tilde{B}:\nabla u\right),
\end{equation}
where $A$, $B$, $\tilde{A}$, and $\tilde{B}$ are the same as for the
BGK equation in the previous section.

Now, we rewrite $L_F(F_1)$ and $L_G(F_1,G_1)$, defined in~\eqref{eq-LFLG}, by using the change of
variables $V = \frac{v-u}{\sqrt{T}}$ and $G= \frac{\eps}{T}$ to get
\begin{equation*}
\begin{split}
& L_F(F_1) = -V\cdot \nabla_V F_1 + \nabla_V \cdot (\nabla_V F_1) 
+ 2\left( (1-J)\partial_J F_1 + J\partial_{JJ}F_1  \right) , \\
& L_G(F_1,G_1) = L_F(G_1) + 2(F_1-G_1).
\end{split}
\end{equation*}
Then simple calculation of derivatives show that $A$, $B$, $\tilde{A}$, and $\tilde{B}$ satisfy the
following properties
\begin{align*}
&   L_F(A) = -3 A, \qquad L_F(B) = -2 B, \\
&   L_G(A,\tilde{A}) = -3 \tilde{A}, \qquad L_G(B,\tilde{B}) = -2 \tilde{B}. 
\end{align*}
Therefore, we look for $F_1$ and $G_1$ as solution
of~(\ref{eq-F1G1FP}) under the following form
\begin{equation*}
  F_1 = a A \cdot \frac{\nabla T}{\sqrt{T}} + b B:\nabla u \quad
  \text{ and } \quad 
   G_1 = \tilde{a} \tilde{A} \cdot \frac{\nabla T}{\sqrt{T}} + \tilde{b} \tilde{B}:\nabla u ,
\end{equation*}
and we find $\tilde{a}=a=-1/3$ and $\tilde{b}=b=1/2$.

Finally, using these relations into $\sigma$ and $q$ and using some
Gaussian integrals (see appendix~\ref{app:CE}) give
\begin{equation*}
  \sigma(M_{vib}F_1) = -\mu \left(\nabla u + (\nabla u)^T
      -\alpha \nabla \cdot u \, I\right) \quad \text{ and } \quad 
   q (M_{vib}F_1,e_{vib}M_{vib}G_1) = -\kappa\nabla_x T,
\end{equation*}
where $\alpha= \frac{c_p}{c_v}-1$, $\mu = \frac{\tau}{2}\rho T$, and
$\kappa = \frac{2}{3}\mu c_p(T)$, which is the announced result, in a
non-dimensional form.

\section{Conclusion}
\label{sec:conclusion}

In this paper, we have proposed to different models (BGK and
Fokker-Planck) of the Boltzmann equation that allow for vibrational
energy discrete modes. These models satisfy the conservation and
entropy property (H-theorem), and the vibration energy variable can be
eliminated by the usual reduced distribution function. The low
complexity of the reduced BGK model can make it attractive to be implemented
in a deterministic code, while the Fokker-Planck model can be easily
simulated with a stochastic method.

Of course, since these models are based on a single time relaxation,
they cannot allow for multiple relaxation times scales. This is not
physically correct, since it is known that the relaxation times for
translational, rotational, and vibrational energies are very
different. However, standard procedures can be used to extend our
model, like the ellipsoidal-statistical approach, already used to
correct the Prandtl number of the BGK model~\cite{esbgk_poly} and
Fokker-Plank models~\cite{Mathiaud2016,Mathiaud2017}.

\appendix
\section{Gaussian integrals and other summation formulas}
\label{app:CE}

In this section, we give some integrals and summation formula that are
used in the paper.

First, we remind the definition of the absolute Maxwellian $M_0(V) =
\frac{1}{(2\pi)^{\frac32}}\exp(-\frac{|V|^2}{2})$. We denote by
$\cint{\phi} = \int_{\R^3}\phi(V)\, dV$ for any function $\phi$. It is standard to
derive the following integral relations (see~\cite{chapmancowling},
for instance), written with the Einstein notation:
\begin{align*}
&   \cint{M_0}_V = 1, \\
&   \cint{V_iV_jM_0}_V = \delta_{ij}, \qquad \cint{V_i^2M_0}_V = 1,
  \qquad \cint{|V|^2M_0}_V = 3, \\
& \cint{V_iV_jV_kV_lM_0}_V = \delta_{ij}\delta_{kl}  +
  \delta_{ik}\delta_{jl}  + \delta_{il}\delta_{jk} , \qquad \cint{V_i^2V_j^2M_0}_V = 1 + 2\, \delta_{ij} \\
& \cint{V_iV_j|V|^2M_0}_V = 5 \,\delta_{ij},  \qquad \cint{|V|^4M_0}_V
  = 15, \\
& \cint{V_iV_j|V|^4M_0}_V = 35 \,\delta_{ij},  \qquad \cint{|V|^6M_0} = 105,
\end{align*}
while all the integrals of odd power of $V$ are zero. \correct{Note that the first relation of each line implies the other relations of the same line: these relations are given here to improve the readability of the paper}.
From the previous Gaussian integrals, it can be shown that for any
$3\times 3$ matrix $C$, we have
\begin{equation*}
\cint{V_iV_jC_{kl}V_kV_lM_0}_V = C_{ij} + C_{ji} + C_{ii}\delta_{ij}.
\end{equation*}

Finally, we have also used the following relations:
\begin{equation*}
  \int_{0}^{+\infty} J e^{-J} \, dJ =   \int_{0}^{+\infty} e^{-J} \, dJ = 1,
\end{equation*}
and also
\begin{equation*}
\sum_{i=0}^{+\infty} e^{-iT_0/T} = \frac{1}{1-e^{-T_0/T}}
 \quad 
\text{ and } \quad 
\sum_{i=0}^{+\infty} i e^{-iT_0/T} = \frac{e^{-T_0/T}}{(1-e^{-T_0/T})^2}.
\end{equation*}

 \bibliographystyle{unsrt}
 \bibliography{biblio}

\begin{thebibliography}{10}

\bibitem{bird}
G.~A. Bird.
\newblock {\em Molecular Gas Dynamics and the Direct Simulation of Gas Flows}.
\newblock Oxford Engineering Science Series, 2003.

\bibitem{BS_2017}
Thomas E.~Schwartzentruber Iain D.~Boyd.
\newblock {\em Nonequilibrium Gas Dynamics and Molecular Simulation}.
\newblock Cambridge Aerospace Series. Cambridge University Press, 2017.

\bibitem{dimarco_pareschi_2014}
G.~Dimarco and L.~Pareschi.
\newblock Numerical methods for kinetic equations.
\newblock {\em Acta Numerica}, 23:369–520, 2014.

\bibitem{luc_2014}
Luc Mieussens.
\newblock A survey of deterministic solvers for rarefied flows (invited).
\newblock {\em AIP Conference Proceedings}, 1628(1):943--951, 2014.

\bibitem{bgk}
E.P. Gross, P.L. Bhatnagar, and M.~Krook.
\newblock A model for collision processes in gases.
\newblock {\em Physical review}, 94(3):511--525, 1954.

\bibitem{Chu_1965}
C.~K. Chu.
\newblock Kinetic-theoretic description of the formation of a shock wave.
\newblock {\em Phys. Fluids}, 8(1):12, 1965.

\bibitem{Struchtrup_moment_book}
H.~Struchtrup.
\newblock {\em Macroscopic Transport Equations for Rarefied Gas Flows
  Approximation Methods in Kinetic Theory}.
\newblock Interaction of Mechanics and Mathematics. Springer, 2005.

\bibitem{cercignani}
C.~Cercignani.
\newblock {\em {T}he {B}oltzmann {E}quation and {I}ts {A}pplications},
  volume~68.
\newblock Springer-Verlag, Lectures Series in Mathematics, 1988.

\bibitem{Grj2011}
M.H. Gorji, M.~Torrilhon, and Patrick Jenny.
\newblock {F}okker-{P}lanck model for computational studies of monatomic
  rarefied gas flows.
\newblock {\em Journal of fluid mechanics}, 680:574--601, August 2011.

\bibitem{Holway}
Jr. Lowell H.~Holway.
\newblock New statistical models for kinetic theory: Methods of construction.
\newblock {\em Physics of Fluids}, 9(9):1658--1673, 1966.

\bibitem{esbgk_poly}
P.~Andries, P.~Le~Tallec, J.-P. Perlat, and B.~Perthame.
\newblock The {G}aussian-{BGK} model of boltzmann equation with small prandtl
  number.
\newblock {\em Eur. J. Mech. B-Fluids}, pages 813--830, 2000.

\bibitem{S_model}
E.~M. Shakhov.
\newblock Generalization of the {K}rook relaxation kinetic equation.
\newblock {\em Izv. Akad. Nauk SSSR. Mekh. Zhidk. Gaza}, 1(5):142--145, 1968.

\bibitem{Grj2013}
M.~Hossein Gorji and Patrick Jenny.
\newblock {A} {F}okker-{P}lanck based kinetic model for diatomic rarefied gas
  flows.
\newblock {\em Physics of fluids}, 25(6):062002--, June 2013.

\bibitem{Mathiaud2016}
J.~Mathiaud and L.~Mieussens.
\newblock A {F}okker{\textendash}{P}lanck model of the {B}oltzmann equation
  with correct {P}randtl number.
\newblock {\em Journal of Statistical Physics}, 162(2):397--414, Jan 2016.

\bibitem{Mathiaud2017}
J.~Mathiaud and L.~Mieussens.
\newblock A {F}okker{\textendash}{P}lanck model of the {B}oltzmann equation
  with correct {P}randtl number for polyatomic gases.
\newblock {\em Journal of Statistical Physics}, 168(5):1031--1055, Sep 2017.

\bibitem{RS_2014}
Behnam Rahimi and Henning Struchtrup.
\newblock Capturing non-equilibrium phenomena in rarefied polyatomic gases: A
  high-order macroscopic model.
\newblock {\em Physics of Fluids}, 26(5):052001, 2014.

\bibitem{WYLX_2017}
Zhao Wang, Hong Yan, Qibing Li, and Kun Xu.
\newblock Unified gas-kinetic scheme for diatomic molecular flow with
  translational, rotational, and vibrational modes.
\newblock {\em Journal of Computational Physics}, 350:237 -- 259, 2017.

\bibitem{ARS_2017}
Takashi Arima, Tommaso Ruggeri, and Masaru Sugiyama.
\newblock Rational extended thermodynamics of a rarefied polyatomic gas with
  molecular relaxation processes.
\newblock {\em Phys. Rev. E}, 96:042143, Oct 2017.

\bibitem{KKA_2019}
S.~Kosuge, Hung-Wen Kuo, and Kazuo Aoki.
\newblock A kinetic model for a polyatomic gas with temperature-dependent
  specific heats and its application to shock-wave structure.
\newblock submitted, 2019.

\bibitem{anderson}
J.~D. Anderson.
\newblock {\em Hypersonic and high-temperature gas dynamics second edition}.
\newblock American Institute of Aeronautics and Astronautics, 2006.

\bibitem{Morse_1964}
T.~F. Morse.
\newblock Kinetic model for gases with internal degrees of freedom.
\newblock {\em Phys. Fluids}, 7(159), 1964.

\bibitem{HH_1968}
A.~B. Huang and D.~L. Hartley.
\newblock Nonlinear rarefied couette flow with heat transfer.
\newblock {\em Phys. Fluids}, 11(6):1321, 1968.

\bibitem{Mathiaud2018}
C.~Baranger, G.~Marois, J.~Mathé, J.~Mathiaud, and L.~Mieussens.
\newblock A {BGK} model for high temperature rarefied gas flows.
\newblock {\em Work in progress}, 2018.

\bibitem{chapmancowling}
S.~Chapman and T.G. Cowling.
\newblock {\em The mathematical theory of non-uniform gases}.
\newblock Cambridge University Press, 1970.

\end{thebibliography}

\end{document}